\documentclass[journal]{IEEEtran}

\usepackage{amssymb}
\usepackage{amsmath,amssymb,amsfonts}
\usepackage{multirow}
\usepackage{graphicx}
\usepackage{textcomp}
\usepackage{amsthm}

\usepackage{setspace}
\usepackage{subfigure}

\usepackage{graphicx}
\usepackage{algorithm}
\usepackage{algorithmic}
\usepackage{amsmath}
\usepackage{color}
\usepackage{amsfonts}

\ifCLASSINFOpdf
\else
\fi

\newtheorem{thm}{Theorem}
\newtheorem{lem}{Lemma}

\newtheorem{defn}{Definition}

\newtheorem{exmp}{Example}

\newtheorem{pro}{Problem}

\begin{document}

\title{Reliable Traffic Monitoring Mechanisms Based on Blockchain in Vehicular Networks}

\author{Jianxiong Guo,
	Xingjian Ding,
	Weili Wu,~\IEEEmembership{Senior Member,~IEEE}
	\thanks{J. Guo and W. Wu are with the Department
		of Computer Science, Erik Jonsson School of Engineering and Computer Science, Univerity of Texas at Dallas, Richardson, TX, 75080 USA; X. Ding is with the School of Information, Renmin University of China, Beijing, CHN
		
		E-mail: jianxiong.guo@utdallas.edu}
	\thanks{Manuscript received April 19, 2005; revised August 26, 2015.}}

\markboth{Journal of \LaTeX\ Class Files,~Vol.~14, No.~8, August~2015}%
{Shell \MakeLowercase{\textit{et al.}}: Bare Demo of IEEEtran.cls for IEEE Journals}

\maketitle

\begin{abstract}
	The real-time traffic monitoring is a fundamental mission in a smart city to understand traffic conditions and avoid dangerous incidents. In this paper, we propose a reliable and efficient traffic monitoring system that integrates blockchain and the Internet of vehicles technologies effectively. It can crowdsource its tasks of traffic information collection to vehicles that run on the road instead of installing cameras in every corner. First, we design a lightweight blockchain-based information trading framework to model the interactions between traffic administration and vehicles. It guarantees reliability, efficiency, and security during executing trading. Second, we define the utility functions for the entities in this system and come up with a budgeted auction mechanism that motivates vehicles to undertake the collection tasks actively. In our algorithm, it not only ensures that the total payment to the selected vehicles does not exceed a given budget, but also maintains the truthfulness of auction process that avoids some vehicles to offer unreal bids for getting greater utilities. Finally, we conduct a group of numerical simulations to evaluate the reliability of our trading framework and performance of our algorithms, whose results demonstrate their correctness and efficiency perfectly.
\end{abstract}


\IEEEpeerreviewmaketitle

\section{Introduction}
\IEEEPARstart{F}{or} the past few years, vehicles have been getting smarter with the progress of technology by installing camera sensors, microcomputers, and communication devices \cite{zhou2014chaincluster} \cite{he2015full}. These vehicles are connected with themselves and other facilities to form a vehicular network. Due to its potential commercial value, the Internet of Vehicles (IoV) has become a hot topic that attracts the attention of academia and industry. The IoV provides a convenient platform for information exchange and sharing between users, such as traffic accidents and road conditions, which can improve the utilization of resources and traffic conditions effectively \cite{zhang2017security}. However, too many mobile and variable entities are involved in the vehicular network, they are usually strangers and do not trust each other. Because of that, information exchange in the IoV still face huge challenges about how to ensure information security and privacy protection in real applications.

The advent of blockchain technology makes it possible to change all of this. Blockchain is a public and distributed database which entered people's sight since Nakamoto published his white paper in 2008 \cite{nakamoto2008bitcoin}. Blockchain takes advantage of the knowledge of modern cryptography and distributed consensus protocol to achieve the purpose of security and privacy protection. Here, the digital signature is an effective means to achieve identity authentication and avoid information leakage, while the consensus protocol allows information to interact freely among users who do not trust each other without the need for a third platform. Thus, it can be used to construct a reliable information trading system because of its decentralization, security, and anonymity \cite{sharma2017software} \cite{li2018blockchain}. Despite this, how to design a reasonable blockchain structure based on the vehicle network, improve the efficiency of information exchange, and guarantee the truthfulness of pricing strategies is still a problem worthy of in-depth discussion.

In this paper, we consider such a scenario: In a smart city, there is a traffic administration (TA) that is responsible for monitoring real-time traffic conditions by collecting information of various road nodes in this city. The traditional approach is to install cameras at various road nodes, but doing so is costly and vulnerable to be damaged by natural and human activities. Especially in some remote corners, there is no need to install a separate camera. Driven by the IoV, the TA can crowdsource its tasks of traffic information collection to vehicles running on the road. Based on that, we need to study the traffic information exchange between TA and vehicles in a smart city. There are two challenges shown as follows: (1) how to construct a reliable and efficient traffic information trading framework implemented by blockchain?; and (2) how to build a fair and truthful mechanism that selects a subset of vehicles as information providers to maximize the TA's profit while preventing vicious competition?

To settle the first challenge, we propose a blockchain-based real-time traffic monitoring (BRTM) system to achieve a reliable and efficient information trading process between TA and vehicles. First, we use modern cryptography methods and digital signature techniques to protect the contents of communication from privacy leakage. In addition to this, the traditional proof-of-work (PoW) consensus mechanism cannot be applied to our resource-limited vehicle network because of its high computational cost and slow confirmation speed \cite{su2018secure}. Thus, we design a lightweight blockchain relied on the reputation-based delegated proof-of-stake (DPoS) consensus mechanism that is able to reduce confirmation time and improve throughput for this system. Here, All the TAs of different cities can be considered as the full nodes which are interconnected with each other to complete the consensus process. There is a reputation value associating with each TA in the BRTM system, which reflects its behaviors in the previous consensus rounds. A full node with a higher reputation value implies it has a higher voting weight and is more likely to become the leader in the future consensus round, thus ensuring the reliability and efficiency of the consensus process.

To settle the second challenge, we propose a budgeted auction mechanism so as to incentivize the vehicles in a city to take on the tasks issued by the TA actively. Here, the TA publishes a task set that contains a number of different tasks about traffic information on various locations. Each active vehicle in this city submits the tasks that it can accomplish and corresponding bid to the TA, then the TA selects a winner set from these active vehicles to undertake its tasks and pay them accordingly. We aim to maximize the TA's profit by selecting an optimal winner set, which can be categorized as a non-monotone submodular maximization problem with knapsack constraint. The simple greedy strategy can give us a valid solution, but it does not satisfy the truthfulness. If no truthfulness, a vehicle may increase its utility by offering a higher bid or colluding with other vehicles, which will damage the fairness and effectiveness of the auction mechanism. Thus, we design a truthful budgeted selection and pricing (TBSAP) algorithm that can guarantee the individual rationality, profitability, truthfulness, and computational efficiency for the budgeted auction mechanism.

Finally, the effectiveness of our proposed mechanisms is evaluated by numerical simulations. It is shown that our proposed mechanisms can enhance the reliability of the blockchain system by urging the nodes to behave rightfully and make sure that the information trading is fair and truthful. To our best knowledge, this is the first time to put forward a reliable blockchain-based traffic monitoring system in a vehicular network that is based on a budgeted auction mechanism to maximize the profit. The rest of this paper is organized as follows: In Sec.  \uppercase\expandafter{\romannumeral2}, we discuss the-state-of-art work. In Sec. \uppercase\expandafter{\romannumeral3}, we introduce our traffic monitoring system and define the utility functions. In Sec. \uppercase\expandafter{\romannumeral4}, we present our design of lightweight blockchain elaborately. In Sec. \uppercase\expandafter{\romannumeral5}, we propose a truthful budgeted auction mechanism. Then we evaluate our proposed system and algorithms by numerical simulations in Sec.\uppercase\expandafter{\romannumeral6} and show the conclusions in Sec. \uppercase\expandafter{\romannumeral7}.

\section{Related Work}
With the development of Internet of Things (IoT), vehicles are no longer isolated individuals, but nodes in a connected network, which lead to the formation of the IoV \cite{kaiwartya2016internet}. Then, the IoV has been developed further due to the progress of computing power and modern communication devices. Kaiwartya \textit{et al.} \cite{kaiwartya2016internet} provided us with an IoV framework that focused on networking architecture mainly. Singh \textit{et al.} \cite{singh2019internet} presented an abstract network model for the IoV and related with different services relied on popular technologies. Ji \textit{et al.} \cite{ji2020survey} designed a novel network structure for the future IoV and gave a comprehensive review of the basic IoV information, which included several network architectures and representative applications of IoV.

In recent years, the rise of blockchain technology overturned the operating mode of the traditional IoT, thereby the blockchain-based IoT has promoted changes in trading methods. Firstly, in energy trading, Li \textit{et al.} \cite{li2017consortium} designed a secure energy trading system, energy blockchain, for the industrial IoT relied on consortium blockchain and credit-based payment scheme. Guo \textit{et al.} \cite{9165852} \cite{guo2020architecture} proposed a blockchain-based distributed multiple energies trading framework to address security and privacy protection. Xia \textit{et al.} \cite{xia2020bayesian} came up with a blockchain-enabled vehicle-to-vehicle electricity trading scheme that exploited Bayesian game pricing to deal with incomplete information sharing. Secondly, in computing resource trading, Yao \textit{et al.} \cite{yao2019resource} studies a resource management problem between the cloud server and miners since the computational power of industrial IoT was limited. Ding \textit{et al.} \cite{ding2020incentive} created a mechanism for the blockchain platform to attract lightweight devices to purchase more computational power from edge servers for participating in mining. Finally, in information trading, Wang \textit{et al.} \cite{wang2017blockchain} designed and realized a securer and more reliable government information sharing system by combining blockchain, network model, and consensus algorithm. Xu \textit{et al.} \cite{xu2018making} devised a blockchain-based big data sharing framework that was appropriate for the resource-limited devices. Chen \textit{et al.} \cite{chen2019secure} proposed a blockchain-based data trading framework that relied on consortium blockchain to achieve secure and truthful data trading for IoV. However, they did not really achieve reliability and ignored the high latency of the consensus process. In this paper, our trading model is different from either of the provious ones and attempt to design a secure and reliable system carefully.

The auction mechanism, as a technique of game theory, has been used to solve the allocation and pricing problems in a variety of applications, such as mobile crowdsensing \cite{zhang2015incentives} \cite{yang2015incentive}, edge computing \cite{kiani2017toward} \cite{jiao2019auction}, and spectrum trading \cite{zhu2014double} \cite{zheng2014strategy}. Zhang \textit{et al.} \cite{zhang2015incentives} surveyed the diverse incentive strategies that encouraged users to take part in mobile crowdsourcing. Yang \textit{et al.} \cite{yang2015incentive} designed two crowdsourcing models under the framework of the Stackelberg game and reverse auction theory respectively. Kaini \cite{kiani2017toward} put forward and solved an auction-based profit maximization problem under the background of the hierarchical model. Jiao \textit{et al.} \cite{jiao2019auction} proposed an auction-based market model for computing resource allocation between the edge server and miners. Zhu \textit{et al.} \cite{zhu2014double} devised a truthful double auction mechanism for cooperative sensing in radio networks. Zheng \textit{et al.} \cite{zheng2014strategy} achieved a social welfare maximization problem by re-distributing wireless channels, while they took the strategy-proof into consideration as well. In this paper, our proposed auction mechanism is similar to that under the mobile crowdsourcing, but there is a budgeted constraint in our model that constrains the range of the winner set. Therefore, both the theoretical analysis and applicable scope are different from the previous work.

\section{Traffic Monitoring System}
In this section, we introduce our traffic monitoring system including network model and utility functions for the entities.

\subsection{Network Model}
Consider a smart city, there is a traffic administration (TA) that is responsible for monitoring real-time traffic condition in this city. It can crowdsource its tasks of traffic information collection to vehicles on the road by paying them a certain amount of money. Therefore, the blockchain-based real-time traffic monitoring (BRTM) system consists of two important entities: TA and vehicle, which is shown in Fig. \ref{fig1}. In the BRTM system, the functionality of the entities in a smart city $S_a$, $a\in\mathbb{Z}^+$, can be shown as follows:

\begin{figure}[!t]
	\centering
	\includegraphics[width=\linewidth]{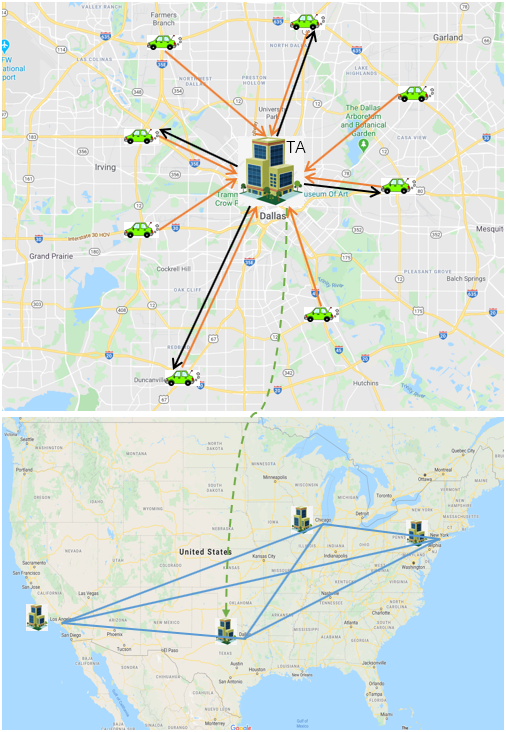}
	\caption{An instance of BRTM system built on the United States, where the orange lines are active vehicles' requests to undertake the tasks and the black line are TA's acceptances for corresponding requests.}
	\label{fig1}
\end{figure}

\textit{Traffic administration (TA): }The TA in the city $S_a$ can be denoted by $TA_a$, which collects the real-time traffic information from those vehicles running in this city. For instance, during peak time, traffic accidents often occur on some heavily trafficked roads. At this time, the $TA_a$ needs to analyze road conditions instantly according to the responses provided by those vehicles that run on the designated location. If $TA_a$ adopts the information provided by some vehicles, it must pay them with data coins. Here, data coin is a kind of digital currency considered as the payment for traffic information.

\textit{Vehicles: }The vehicles in this system are installed with camera sensors, wireless communication devices, and microcomputer systems, which are able to collect, process, and transmit data to other corresponding devices. The active vehicles in the city $S_a$ can be denoted by a set $V_a=\{v_{a1},v_{a2},\cdots,v_{ab},\cdots\}$, where active vehicles are those participating in information trading with $TA_a$ and idle vehicles do not want to share their traffic information. Once a transaction between $TA_a$ and $v_{ab}$ has been completed and validated, the vehicle $v_{ab}$ will receive a certain number of data coins from $TA_a$.

Based on that, a smart city in our BRTM system can be denoted by $S_a=(TA_a, V_a)$, and the whole BRTM system $\mathbb{S}$ can be denoted by $\mathbb{S}=\{S_1,S_2,\cdots,S_a,\cdots\}$. It is very easy to understand, for example, this system can be constructed in a country that is composed of a number of smart cities. Here, all TAs that serve under the cities in $\mathbb{S}$ are interconnected with each other to make up a peer-to-peer (P2P) wide-area network. We take Fig. \ref{fig1} as an instance to demonstrate it.
\begin{exmp}
	Shown as the upper part in Fig. \ref{fig1}, there is a TA in the Dallas city that announces a number of traffic information collection tasks across the important intersections of the whole city. Some vehicles in this city request to undertake part of tasks at a certain price. For a vehicle, whether its request can be accepted is decided by the TA according to the TA's evaluation. Then, a BRTM system is built on the United States, shown as the lower part in Fig. \ref{fig1}. The four main cities, Los Angeles, Chicago, New York, and Dallas are interconnected with each other to form a large network.
\end{exmp}

The P2P network among the TAs in the BRTM system is used for constructing the blockchain network. Here, blockchain is an effective technique to record the transactions and guarantee the security. In our blockchain network, there are two types of nodes which are full node and lightweight node. For the full nodes, they not only executes the consensus process to validate the candidate block but also manages stores the blockchain that contains the whole transactions between TA and vehicles. For the lightweight nodes, they merely store the block headers that is convenient for them to check and verify. In our system $\mathbb{S}$, each $TA_a$ of $S_a\in\mathbb{S}$ is defaulted as a full node and each vehicle $v_{ab}\in V_a$ of $S_a\in\mathbb{S}$ is defaulted as a lightweight node because of lacking enough storage and computational power. Details of the blockchain design will be described in Sec. \uppercase\expandafter{\romannumeral4}.

\subsection{Definitions of Utility functions}
Consider a smart city $S=(TA,V)$ such that $S\in\mathbb{S}$ at some point, we neglect the subscripts and denote by $V=\{v_1,v_2,\cdots,v_n\}$ for convenience. This traffic administration $TA$ releases a task set $T=\{t_1,t_2,\cdots,t_m\}$ about traffic information for the active vehicles to undertake. For the $TA$, there is an appraisement $a_j$ for each task $t_j\in T$. For each vehicle $v_i\in V$, it can attempt to finish a subset of tasks $T_i\subseteq T$ based on its ability and willingness. The cost of $T_i$ for the $v_i$ can be defined as $c_i$, which is private and not known to the $TA$. Then, the vehicle $v_i$ determines a bid $b_i$ and respones the $TA$ with its bid-task pair $(T_i,b_i)$, where $b_i$ is $v_i$'s reserve price that is the lowest price it want to undertake the tasks $T_i$. Generally, we have $T_i\cap T_j\neq\emptyset$ for any vehicle $v_i$ and $v_j$. Here, for each task $t_i\in T$, we can quantify it as the traffic information at some position in this city, thereby an active vehicle is able to submit a subset of tasks based on its running path. After receiving the responses from all active vehicles, the $TA$ needs to select a winner (acceptance) set $W\subseteq V$ and give a payment $p_i$ for each $v_i\in W$. The utility for a vehicle $v_i\in V$ is
\begin{equation}
U_i=\left
\{\begin{IEEEeqnarraybox}[\relax][c]{l's}
p_i-c_i,&if $v_i\in W$\\
0,&if $v_i\notin W$
\end{IEEEeqnarraybox}
\right.
\end{equation}
The profit of the $TA$ is
\begin{equation}
P(W)=A(W)-\sum\nolimits_{v_i\in W}p_i
\end{equation}
where we denote by $A(W)=\sum_{t_j\in\cup_{v_i\in W}T_i}a_j$. Besides, in order to control cost, the $TA$ has a total budget $B$ such that $\sum_{v_i\in W}b_i\leq B$. Based on that, we want to design an incentive mechanism to execute this auction process.

\begin{figure}[!t]
	\centering
	\includegraphics[width=\linewidth]{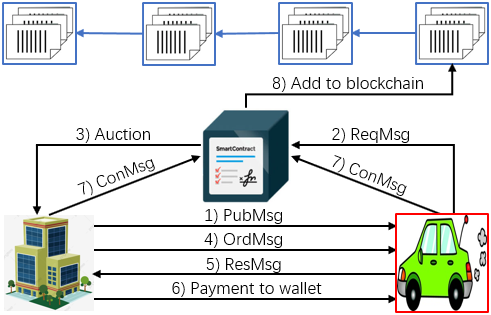}
	\caption{The traffic information trading framework between TA and vehicles.}
	\label{fig2}
\end{figure}

\section{Lightweight Blockchain Design}
In a smart city, the TA publishs a task set first, then those active vehicles in this city request to undertake part of tasks. If a vehicle's request is accepted by the TA, it will get a certain number of rewards. With the help of the blockchain technology, the TA can trade the traffic information with the active vehicles in a decentralized, secure, and verifiable manner. In the meantime, a budgeted auction mechanism is written in the smart contract. When a new round of trading begins, the smart contract in our BRTM system will be deployed and executed automatically in order to find the optimal profit for the $TA$ according to the vehicles' requests. 

\subsection{System Initialization}
Consider a smart city $S=(TA,V)$, each vehicle $v_i\in V$ has to register with the certificate authority (CA) through correlating its license plate so as to acquire a unique identification $ID_i$. Then, each legitimate vehicle will be assigned with a private/public key pair $(SK_i,PK_i)$, where the private key is reserved by itself and the public key is known by all the legitimate nodes in the blockchain as a pseudonym. The asymmetric encryption are adopted widely in the current blockchain technology for the sake of data integrity. The data encrypted by the private (resp. public) key can be decrypted by the public (resp. private) key. After registering, the vehicle $v_i$ will obtain a certificate $Cer_i$ signed by CA's private key that can certify the authenticity of its identity. There is an account $A_i=\{ID_i,Cer_i,(SK_i,PK_i),B_i\}$ associated with the $v_i\in V$ where $B_i$ is its data coin balance. Similarly, the $TA$ has an account denoted by $A_*=\{ID',Cer',(SK',PK'),B',Rep'\}$ as well, where $Rep'$ is its reputation value which will be introduce later. If a round of trading is finished successfully, data coins will be transferred from $B_*$ to $B_i$.

\subsection{Information Trading Framework}
The elaborate process of traffic information trading in our BRTM system are drawn in Fig. \ref{fig2}, whose specific operation steps are summarized as follows:

\textbf{1) }The $TA$ publishs its task set $T$ across the active vehicles in this city. The published message can be denoted by $PubMsg=\{SK'(T),Cer',STime\}$, where is the task set is encrypted by the $TA$'s private key for security and $STime$ is the time stamp of this message generation.

\textbf{2) }After receiving the $PubMsg$ from the $TA$, each active vehicle $v_i\in V$ in this city has to verify the $TA$'s certification $Cer'$ by the CA's public key and then read the task set by the $TA$'s public key. Once verified, each $v_i$ will determine its task-bid pair $(T_i,b_i)$ and send its request message $ReqMsg$, denoted by $ReqMsg=\{PK'(T_i,b_i),Cer_i,STime\}$, back to the $TA$. Here, its task-bid pair is encrypted by the $TA$'s public key because it can only be read by the $TA$ for privacy.

\textbf{3) }The $TA$ waits to collect all the $ReqMsg$ from the active vehicles in this city, then determine the winner set $W$ and its corresponding payment $p_i$ for $v_i\in W$. This process is executed by the built-in smart contract, which adopts a budgeted auction mechanism explained in Sec. \uppercase\expandafter{\romannumeral5}.

\textbf{4) }The $TA$ sends the order message to each accepted vehicle $v_i\in W$, which can be denoted by $OrdMsg=\{PK_i(p_i),Cer',STime\}$. Here, the payment is encrypted by the $v_i$'s public key which can be read by the $v_i$ merely.

\textbf{5) }If receiving the $OrdMsg$ from the TA, it means that the $v_i$ has been selected as an information provider. Then, the $v_i$ can read it by its private key and send the response message that includes the traffic information associated with task $T_i$, denoted by $ResMsg=\{PK'(data),Cer_i,STime\}$, back to the $TA$. The data is encrypted by the $TA$'s public key due to the same reason as (2).

\textbf{6) }After the data transmission, the $TA$ will check and confirm whether the traffic information about task $T_i$ from the $v_i$ meets its requirement. If yes, the $TA$ pays the data coins $p_i$ to the $v_i$'s public wallet address.

\textbf{7) }Once receiving the payment from the $TA$, the $v_i$ will sends the confirm message, denoted by $ConMsg=\{PK'(Confirm),Cer_i,STime\}$, back to the $TA$. The $TA$ will sign it with a confirm message as well, then a transaction record is formulated until now.

\textbf{8) }The built-in smart contract in the $TA$ will record a series of transactions between $TA$ and vehicles in this city within a period of time, and package them into a new block. Before this block is added into the blockchain, it has to be reached a consensus among all the TAs in the BRTM system. The consensus process we adopt here is called as reputation-based delegated proof-of-work (DPoS) mechanism explained in Sec. \uppercase\expandafter{\romannumeral4}.C. After finishing the consensus process, this block will be appended into the blockchain in a linear and chronological order permanently.

\subsection{Reputation-based DPoS Mechanism}
A consensus process is essential to guarantee the consistency and security for the blockchain system. The classic PoW consensus mechanism is not suitable to our system because of lacking the necessary computational power in the vehicular networks. Therefore, we design a novel reputation-based DPoS mechanism that not only ensures reliability but also improves performance. Just as said before, all the TAs in the BRTM system $\mathbb{S}$ are considered as the full nodes, thereby the consensus process can be carried out among them. The operational steps are displayed as follows:

\textbf{1) Witness election: }Denote by the full node set $\mathbb{Z}=\{TA_1,TA_2,\cdots,TA_a,\cdots\}$, there is a reputation value $Rep_a$ associated with each full node $TA_a\in\mathbb{Z}$. This reputation $Rep_a$ can be considered as its stake that determines its voting weight directly. It implies that the nodes with higher reputation values can determine the results more than those with lower reputation values. At each witness election epoch, the voting result $R_a$ for each $TA_a\in\mathbb{Z}$ is defined as
\begin{equation}
	R_a=\sum\nolimits_{TA_{a'}\in\mathbb{Z}\backslash\{TA_{a}\}}Rep_{a'}\cdot\mathbb{I}(a,a')
\end{equation}
where the $\mathbb{I}(a,a')$ is an indicator. Here, $\mathbb{I}(a,a')=1$ if the $TA_{a'}$ votes to support $TA_a$, otherwise $\mathbb{I}(a,a')=0$. All these reputation values and voting results are stored in the blockchain, thereby each node can check them by itself. We select a witness committee $\mathbb{M}\subseteq\mathbb{Z}$ that have the top $|\mathbb{M}|$ highest voting results. Then, it can be divided into two part such that $\mathbb{M}=\mathbb{D}\cup\mathbb{E}$ and $\mathbb{D}\cap\mathbb{E}=\emptyset$, where the full node in $\mathbb{D}$ is an active witness and in $\mathbb{E}$ is a standby witness. The active witnesses have the top $|\mathbb{D}|$ highest voting results, which are able to generate a new block like a leader. However, the standby witnesses can only verify and broadcast the block generated by an active witness.

\textbf{2) Block generation: }First, the active witnesses in the set $\mathbb{D}$ are sorted in a random sequence. Then, each $TA_a\in\mathbb{D}$ takes turn to be a leader round by round based on this sorting that is responsible for generating a new block. The leader has to verify and check all the transactions occurred recently and package those valid transactions into a new block. If a node in its turn does not produce a block over a given period of time successfully, it will be skipped and have no chance to be leader again in this election epoch. After passing $|\mathbb{D}|$ consensus rounds, namely an election epoch, we re-execute the witness election to elect a new witness committee $\mathbb{M}$ according to the new reputation values of the full nodes in $\mathbb{Z}$.

\textbf{3) Consensus process: }After producing a new block, the leader should broadcast it to all witnesses in $\mathbb{M}$. All witnesses should verify the leader's identity and check this block, then broadcast their verification results signed by their private key. Next, each witness $TA_a\in\mathbb{M}$ compares its verified outcome with those from other witnesses. If confirmed, it will send a confirmation message to the leader, which means that this witness agrees to accept the new block. For the leader, it will append the new block into the blockchain if the proportion of witnesses that agree to accept this block is more than two thirds. All the full nodes should synchronize their local blockchain storage according to the longest chain principle and all the lightweight node should update their blockchain headers based on the newest one.

\textbf{4) Reputation update: }When a consensus round is finished or interrupted, we need to update the reputation values of the full nodes in $\mathbb{Z}$ according to their behaviors during this consensus round. Generally, the positive behaviors contribute to the accumulation of reputation, but the negative behaviors had a bad effect. We denote by $Rep_a(i)$ the repuation value of the full node $TA_a\in\mathbb{Z}$ at the $i$-th consensus round. At the $i$-th consensus round, we define
\begin{equation}
	\Delta_a(i)=A\cdot\alpha_a(i)+B\cdot\beta_a(i)+C\cdot\gamma_a(i)
\end{equation}
for each full node $TA_a\in\mathbb{Z}$. Here, if the $TA_a\in\mathbb{Z}$ participates the voting in witness election, we give $\alpha_a(i)=1$, otherwise $\alpha_a(i)=-1$. When the $TA_a$ is the leader at this consensus round, we give $\beta_a(i)=1$ if it generates a new block that is accepted by other witnesses eventually, otherwise $\beta_a(i)=-1$. If the $TA_a\in\mathbb{Z}\backslash\{\text{leader}\}$ is not the leader, we give $\beta_a(i)=0$. When the $TA_a\in\mathbb{M}\backslash\{\text{leader}\}$ is a member of witness committee but not the leader, we give $\gamma_a(i)=1$ if it verifies the block produced by the leader correctly, otherwise $\gamma_a(i)=-1$. If the $TA_a\in\mathbb{Z}\backslash\mathbb{M}$ is not a witness, we give $\gamma_a(i)=0$. Then, the $A$, $B$, and $C$ are three adjustable parameters that represent the rewards or punishments of voting, leader, and verification behaviors. Based on its behaviors $\Delta_a(i)$ at the $i$-th consensus round, we define the reputation value of the full node $TA_a\in\mathbb{Z}$ as
\begin{equation*}
Rep_a(i)=\left
\{\begin{IEEEeqnarraybox}[\relax][c]{l's}
\max\{1,Rep_a(i-1)+\Delta_a(i)\},&if $\Delta_a(i)\geq 0$\\
\min\{0,Rep_a(i-1)+\Delta_a(i)\},&if $\Delta_a(i)< 0$
\end{IEEEeqnarraybox}
\right.
\end{equation*}
where the $TA_a$'s voting is meaningless when its reputation down to zero according to the (3). Moveover, we initialize the reputation $Rep_a(0)=0.5$ for each $TA_a\in\mathbb{Z}$ and make $B>C>A>0$ because of their importance.

\subsection{Performance and Reliability}
Our BRTM system inherits the advantages of blockchain technology which helps to ensure the reliability of traffic information trading and privacy protection. Its main characteristics are summarized as follows: (1) Decentralization: the traffic information trading can be performed in a distributed manner without relying on the third trusted intermediaries, which avoids single point failure; (2) Reliability: when there is a node failing because of malicious attacks, the blockchain network can be guaranteed to work as usual; (3) Privacy protection: shown as Sec. \uppercase\expandafter{\romannumeral4}.B, asymmetric encryption is used to prevent the message from reading by others, which makes the private message remain confidential during the auction process; (4) Efficiency: our reputation-based DPos mechanism not only makes the consensus process more reliable but also improves efficiency by screening out the trustworthy nodes in advance; (5) Transaction authentication: all transaction must be checked and verified by the witness before appending into the blockchain, thereby it is extremely hard to dominate the majority of witnesses to create an unreal block. Based on that, our BRTM system provides us with a secure, reliable, and efficient traffic information trading platform.

\section{Budgeted Auction Incentive mechanism} 
In this section, we consider how to model the trading process between $TA$ and $v_i\in V$ in a smart city $S=(TA,V)$. Auction mechanism is an execellent theoretical tool for this scenario. Here, the sellers are $V=\{v_1,v_2,\cdots,v_n\}$ and the buyer is $TA$. Each seller $v_i$ submits a task-bid pair $(T_i,b_i)$, and then the buyer determine whether to accept its request. If accept it, the buyer need to offer a payment. We call such an auction as ``reverse auction''. Thus, the single round auction mechanism can be formulated as Problem 1.
\begin{pro}
	The active vehicles in $S$ given a task-bid pair vector $((T_1,b_1),(T_2,b_2),\cdots,(T_n,b_n))$, the $TA$ need to compute an allocation $W$ and payments $\{p_i\}_{v_i\in W}$ such that
	\begin{equation}
		\max_{W\subseteq V}P(W) \text{ s.t. } \sum\nolimits_{v_i\in W}b_i\leq B
	\end{equation}
	where it satisfies individual rationality, profitability, truthfulness, and computational efficiency.
\end{pro}
\noindent
The auction we use in this section is a single sealed-bid auction, thereby the task-bid pair given by each seller is private and not known to other sellers in order to avoid collaborating or forming an alliance. Once submitted, any seller can not change its task or bid during the auction. This is why we use asymmetric encryption in Sec. \uppercase\expandafter{\romannumeral4}.B.

\begin{defn}[Individual rationality] 
	A reverse auction is individual rational if there is no winning sellers getting less than its cost.
\end{defn}
For each vehicle $v_j\in V$, it must have an non-negative utility that is $p_i\geq c_i$ if $v_i\in W$.
\begin{defn}[Profitability]
	The profit of an reverse auction is the difference between the value generated by winner set and the payment to the sellers, which is nonnegative.
\end{defn}
For the $TA$ in $S$, it has an allocation $W$ such that $\sum\nolimits_{t_j\in\cup_{v_i\in W}T_i}a_j\geq \sum\nolimits_{v_i\in W}p_i$.
\begin{defn}[Truthfulness]
	A reverse auction is truthful if every seller's bid, which is the same as its truthful cost, is the dominant strategy that maximizes its utility.
\end{defn}
For each vehicle $v_i\in V$, it means that the $v_i$ cannot increase its utility by giving a bid $b_i$ that is different from its truthful cost $c_i$ no matter what others' bids are. Truthfulness is a very important property for an auction mechanism to avoid malicious price manipulation as well as ensure a fair and benign market competition environment. In our case, if a vehicle can obtain a better utility when offering an untruthful bid, then those malicious vehicles are able to cheat in the auction that benefit themselves but hurt the interests of others. By making all active vehicles bidding truthfully, the $TA$ can allocate its tasks to the most suitable vehicles. Therefore, truthfulness plays a very significant role in our auction design.
\begin{defn}[Computational efficiency]
	A reverse auction is efficient if it can be done in ploynomial time.
\end{defn}

\subsection{Submodular Maximization}
Assume the $TA$ gives a payment $p_i$ with $p_i=b_i$ for each vehicle $v_i\in V$, the Problem 1 can be reduced to an optimization problem, called ``budgeted vehicle allocation problem'', that selects a winner set $W\subseteq V$ such that $\bar{P}(W)$ is maximized. We denote
\begin{equation}
	\bar{P}(W)=A(W)-\sum\nolimits_{v_i\in W}b_i
\end{equation}
where $\sum_{v_i\in W}b_i\leq B$. To be meaningful, we assume $\bar{P}(\emptyset)=0$ and it exists at least one vehicle $v_k\in V$ such that $\bar{P}(\{v_k\})>0$ with $b_k\leq B$. However, the budgeted vehicle allocation problem can be reduced to the classic set cover problem in polynomial time, thereby it is NP-hard to find the optimal solution definitely. Based on that, we have to seek help from designing an approximation algorithm which will make use of submodularity of objective shown in (6).
\begin{defn}[Submodular function \cite{lovasz1983submodular}]
	Given a set function $f:2^V\rightarrow\mathbb{R}$, it is submodular if
	\begin{equation}
		f(X\cup\{u\})-f(X)\leq f(Y\cup\{u\})-f(Y)
	\end{equation}
	for any $X\subseteq Y\subseteq V$ and $u\in V\backslash Y$.
\end{defn}
\begin{lem}
	The objective function of budgeted vehicle allocation problem $\bar{P}$ is submodular.
\end{lem}
\begin{proof}
	According to the (6), to show $\bar{P}(X\cup\{v_i\})-\bar{P}(X)\leq \bar{P}(Y\cup\{v_i\})-\bar{P}(Y)$ is equivalent to show $A(X\cup\{v_i\})-A(X)\leq A(Y\cup\{v_i\})-A(Y)$ for any $X\subseteq Y\subseteq V$ and $v_i\in V\backslash Y$. We have $A(X\cup\{v_i\})-A(X)=$
	\begin{flalign}
		&=\sum\nolimits_{t_i\in T_i\backslash\cup_{v_j\in X}T_j}a_i\geq\sum\nolimits_{t_i\in T_i\backslash\cup_{v_j\in Y}T_j}a_i\\
		&=A(Y\cup\{v_i\})-A(Y)
	\end{flalign}
	Therefore, $\bar{P}$ is submodular.
\end{proof}

As we known, Buchbinder \textit{et al.} \cite{buchbinder2015tight} designed a double greedy algorithm for the unconstrained submodular maximization problem with a $(1/3)$-approximation under the deterministic setting and a $(1/2)$-approximation under the randomized setting. Back to our budgeted vehicle allocation problem, there is a constraint $\sum_{v_i\in W}b_i\leq B$ that is a knapsack constraint. This constraint affects the resulting profit obtained by the $TA$ and the number of vehicles that are selected as winners in the auction. Thus, a valid reverse auction mechanism design are required to ensure the truthfulness carefully. Based on the aforementioned analysis, this optimization problem can be catagorized to a non-monotone submodular maximization with knapsack constraint. Lee \textit{et al.} \cite{lee2009non} proposed a $(1/5-\epsilon)$-approximation algorithm to maximize any non-negative submodular function with knapsack constraint by means of fractional relaxation and local search method. Even though this algorithm can give us a constant approximation ratio, its process is complex and its performance is worse than the greedy-heuristic algorithm actually. The greedy-heuristic algorithm is shown in Algorithm \ref{a1}.

Here, we denote by $\bar{P}(v_i|W)=\bar{P}(W\cup\{v_i\})-\bar{P}(W)$ and $b^*$ is the bid offered by vehicle $v^*$. Shown as Algorithm \ref{a1}, we select a vehicle $v^*$ from the current constrained set $F$ with maximum unit marginal gain until $\bar{P}(v^*|W)<0$ or $F=\emptyset$. Despite the winner set returned by Algorithm \ref{a1} has no any theoretical bounds because of non-monotonicity and knapsack constraint, it is intuitive and has a good performance in the practical applications. Now, we have to explore whether the greedy-heuristic satisfies the aforementioned four properties in the reverse auction mechanism.

\begin{algorithm}[!t]
	\caption{\text{greedy-heuristic}}\label{a1}
	\begin{algorithmic}[1]
		\renewcommand{\algorithmicrequire}{\textbf{Input:}}
		\renewcommand{\algorithmicensure}{\textbf{Output:}}
		\REQUIRE Function $\bar{P}$, bids $\{b_i\}_{v_i\in V}$, and budget $B$
		\ENSURE Winner set $W$
		\STATE Initialize: $W\leftarrow\emptyset$, $s\leftarrow 0$
		\STATE Initialize: $F\leftarrow\{v_i:v_i\in V\backslash W \text{ and } b_i\leq B\}$
		\WHILE {$F\neq\emptyset$}
		\STATE Select $v^*$ such that $v^*\in\arg\max_{v_i\in F}\{\bar{P}(v_i|W)/b_i\}$
		\IF {$\bar{P}(v^*|W)<0$}
		\STATE Break
		\ENDIF
		\STATE $W\leftarrow W\cup\{v^*\}$
		\STATE $s\leftarrow s+b^*$
		\STATE $F\leftarrow\{v_i:v_i\in V\backslash W \text{ and } b_i\leq B-s\}$
		\ENDWHILE
		\RETURN $W$
	\end{algorithmic}
\end{algorithm}

Next, let us look at whether the greedy-heuristic satisfies profitability, individual rationality, truthfulness, and computational efficiency one by one as follows:

\textbf{1) Individual rationality: }The $TA$ pay each vehicle in the winner set its bid, thus it is individually rational.

\textbf{2) Profitability: }Form line 5 in Algorithm \ref{a1}, the while loop is terminated when there is no vehicle having postive marginal gain, which guarantee $\bar{P}(W)>0$.

\textbf{3) Truthfulness: }Let us consider an example. There are five tasks issued by the $TA$, denoted by $T=\{t_1,t_2,t_3,t_4,t_5\}$, with appraisement $\{a_1=2,a_2=3,a_3=4,a_4=2,a_5=5\}$ and the budget $B=5$. There are three active vehicles denoted by $V=\{v_1,v_2,v_3\}$ in this city. They want to undertake the $TA$'s tasks as $T_1=\{t_1,t_3,t_5\}$, $T_2=\{t_1,t_2,t_5\}$, $T_3=\{t_3,t_4,t_5\}$, $c_1=2$, $c_2=2$, and $c_3=2$. When each vehicle offers a bid truthfully, we have $\bar{P}(\{v_1\}|\emptyset)/b_1=(A(\{v_1\})-b_1)/b_1=(11-2)/2=4.5$. Similarly, we have $\bar{P}(\{v_2\}|\emptyset)/b_2=4$ and $\bar{P}(\{v_3\}|\emptyset)/b_3=3.5$. Thus, vehicle $v_1$ is selected at the first iteration. At the second iteration, we have $\bar{P}(v_2|\{v_1\})/b_2=0.5$ and $\bar{P}(v_3|\{v_1\})/b_3=0$. Thus, vehicle $v_2$ is selected at the second iteration. Then, the greedy-heuristic terminates because the budget is exhausted. However, when vehicle $v_2$ offer an untruthful bid $b_2=c_2+\lambda$, we have $\bar{P}(v_2|\{v_1\})/b_2=3/(2+\lambda)-1$. If $\lambda\in(0,1)$, it will be selected as the winner at the second iteration. The algorithm terminates here and vehicle $v_2$ can get more payment by offering an untruthful bid, thus is does not satisfy truthfulness.

\textbf{4) Computational efficiency: }The running time of the greedy-heuristic is $O(Bnm/\min_{v_i\in V}\{b_i\})$ because it takes $O(nm)$, $|T|=m$ and $|V|=n$, to compute $\bar{P}$ and iterates at most $B/\min_{v_i\in V}\{b_i\}$ times, so computationally efficient.

\begin{algorithm}[!t]
	\caption{\text{TBSAP Algorithm}}\label{a2}
	\begin{algorithmic}[1]
		\renewcommand{\algorithmicrequire}{\textbf{Input:}}
		\renewcommand{\algorithmicensure}{\textbf{Output:}}
		\REQUIRE Function $\bar{P}$, bids $\{b_i\}_{v_i\in V}$, and budget $B$
		\ENSURE Winner set $W$
		\STATE // Winner allocation stage
		\STATE Initialize: $X\leftarrow\emptyset$, $s\leftarrow 0$
		\WHILE {$X\neq V$}
		\STATE Select $v^*$ such that $v^*\in\arg\max_{v\in V\backslash X}\{\hat{P}(v|X)\}$
		\IF {$\hat{P}(v^*|X)<0$ or $s+b^*>B$}
		\STATE Break
		\ENDIF
		\STATE $X\leftarrow X\cup\{v^*\}$
		\STATE $s\leftarrow s+b^*$
		\ENDWHILE
		\STATE // Payment determination stage
		\FOR {each $v_i\in X$}
		\STATE Initialize: $Y_0\leftarrow\emptyset$, $s\leftarrow 0$, $j\leftarrow 0$
		\STATE Initialize: $p_i\leftarrow -\infty$
		\WHILE {$Y_j\neq V_{-i}$}
		\STATE $v_{i_{j+1}}\leftarrow\arg\max_{v\in V_{-i}\backslash Y_j}\{\hat{P}(v|Y_j)\}$
		\IF {$\hat{P}(v_{i_{j+1}}|Y_j)<0$ or $s+b_{i}>B$}
		\STATE Break
		\ENDIF
		\STATE $Y_{j+1}\leftarrow Y_j\cup\{v_{i_{j+1}}\}$
		\STATE $s\leftarrow s+b_{i_{j+1}}$
		\STATE $p_i\leftarrow\max\left\{p_i,b_{i_{j+1}}\cdot\frac{A(v_i|Y_j)}{A(v_{i_{j+1}}|Y_j)}\right\}$
		\STATE $j\leftarrow j+1$
		\ENDWHILE
		\IF {$s+b_i\leq B$}
		\STATE $p_i\leftarrow\max\{p_i,A(v_i|Y_j)\}$
		\ENDIF
		\ENDFOR
		\RETURN $X$ and $\{p_i\}_{v_i\in X}$
	\end{algorithmic}
\end{algorithm}

\subsection{Truthful Auction Mechanism Design}
As mentioned above, the greedy-heuristic algorithm is not meaningful due to lacking truthfulness even though it simple to implement. Thus, we have to design a budgeted reverse auction mechanism that not only gets a good profit by encouraging vehicles to undertake the tasks, but also satisfies the four properties shown as before, especially for truthfulness, in order to protect from manipulating this system malignantly by offering a unreal bid. Therefore, a valid auction mechanism needs to be designed although it may lose some of its profits. We propose a truthful budgeted reverse auction mechanism based on Myerson's introduction \cite{nisan2007}.
\begin{thm}[\cite{nisan2007}]
	An reverse auction mechanism is truthfully if and only if it satisfies as follows: (1) Monotonicity: If a vehicle (seller) $v_i\in V$ wins by its task-bid pair $(T_i,b_i)$, then it will win by any bid that is smaller than $b_i$ with the same task set $T_i$ as well. Namely, the $(T_i,b^\circ_i)$ will win by any bid $b^\circ_i<b_i$ when other sellers do not change their strategies; (2) Critical payment: The payment $p_i$ of a winner $v_i$ with its task-bid pair $(T_i,b_i)$ is the maximum bid with which the $v_i$ can win. Namely, the $(T_i,b^\circ_i)$ will not win by any bid $b^\circ_i>p_i$ when other sellers do not change their strategies.
\end{thm}
Based on the Theorem 1, we design our budgeted reverse auction mechanism consisted of the winner allocation stage and payment determination stage. From the bids $\{b_i\}_{v_i\in V}$, we denote by the unit marginal gain $\hat{P}$
\begin{equation}
	\hat{P}(v_i|X)={\bar{P}(v_i|X)}/{b_i}=[{\bar{P}(X\cup\{v_i\})-\bar{P}(X)}]/{b_i}
\end{equation}
The winner allocation stage selects vehicles from $V$ in a greedy approach. All active vehicles in $V$ are sorted in a non-increasing sequence as
\begin{equation}
	\hat{P}(v_1|X_0)\geq\cdots\geq\hat{P}(v_{i}|X_{i-1})\geq\cdots\geq\hat{P}(v_{n}|X_{n-1})
\end{equation}
where the $i$-th vehicle has the maximum unit marginal gain $\hat{P}(v_{i}|X_{i-1})$ over $V\backslash X_{i-1}$, $X_{i-1}=\{v_1,v_2,\cdots,v_{i-1}\}$, and $X_0=\emptyset$. From this sorting, we select the maximum $L$ with $L\leq n$ such that $\hat{P}(v_L|X_{L-1})\geq 0$ and $\sum_{v_i\in X_L}b_i\leq B$ where $X_L=\{v_1,v_2,\cdots,v_L\}$.

According to the winner set $X_L$, we have to determine the payment price $p_i$ for each vehicle $v_i\in X_L$. The reverse auction mechanism runs the winner allocation process repreatedly. Consider the vehicle $v_i\in V$, all active vehicles in $V_{-i}=V\backslash\{v_i\}$ are sorted in a non-increasing order as follows:
\begin{equation}
\hat{P}(v_{i_1}|Y_0)\geq\hat{P}(v_{i_2}|Y_{1})\geq\cdots\geq\hat{P}(v_{i_{n-1}}|Y_{n-2})
\end{equation}
where the $i_j$-th vehicle has the maximum unit marginal gain $\hat{P}(v_{i_j}|Y_{j-1})$ over $V\backslash Y_{j-1}$, $Y_{j-1}=\{v_{i_1},v_{i_2},\cdots,v_{i_{j-1}}\}$, and $Y_0=\emptyset$. From this sorting, we select the maximum $L_i$ with $L_i\leq n-1$ such that $\hat{P}(v_{i_{L_i}}|Y_{L_i-1})\geq 0$ and $\sum_{v_{i_j}\in Y_{L_i-1}}b_{i_j}+b_i\leq B$ where $Y_{L_i-1}=\{v_{i_1},v_{i_2},\cdots,v_{i_{L_i-1}}\}$. In short, there is a list $Y_{L_i}=\{v_{j_1},v_{j_2},\cdots,v_{i_{L_i}}\}$ associated with the vehicle $v_i$. For each position $i_{j+1}$ in this list $Y_{L_i}$, we can get the maximum bid $b_{i(j+1)}'$ that the vehicle $v_i$ should offer in order to replace $v_{i_{j+1}}$ with $v_i$ at the position $i_j$. To achieve it, we have $\hat{P}(v_i|Y_j)\geq\hat{P}(v_{i_{j+1}}|Y_j)$, which is equivalent to $[A(v_i|Y_j)-b_i]/b_i\geq[A(v_{i_{j+1}}|Y_j)-b_{i_{j+1}}]/b_{i_{j+1}}$. Thus, we have the following inequality:
\begin{equation}
	b_{i(j+1)}'= b_{i_{j+1}}\cdot{A(v_i|Y_j)}/{A(v_{i_{j+1}}|Y_j)}
\end{equation}
From here, we can achieve a list $b'_i=\{b'_{i(1)},b'_{i(2)},\cdots,b'_{i(L_i)}\}$ where each $b'_{i(j)}$ is the maximum bid to replace the vehicle $v_{i_j}$ with $v_i$ in the list $Y_{L_i}$. The vehicle $v_i$ can replace any one in $Y_{L_i}$ without exceeding the budget $B$. Consider the $v_{i_{L_i+1}}$, it can be divided into the following three cases: (1) $\sum_{v_{i_j}\in Y_{L_i}}b_{i_j}+b_i> B$; (2) $\hat{P}(v_{i_{L_i+1}}|Y_{L_i})< 0$; and (3) $L_i=n-1$ where $v_{i_{L_i+1}}$ does not exist and can be considered as a virtual vehicle. For the case (1), the $v_{i_{L_i+1}}$ cannot be replaced by the $v_i$ because of the budget constraint, thereby we do not need to do anything. If not case (1), for the cases (2) and (3), the bid by the $v_i$ should be less than $A(v_i|Y_{L_i})$ so as to replace the $v_{i_{L_i+1}}$, thereby the maximum bid to replace the $v_{i_{L_i+1}}$ with $v_i$ is equal to $A(v_i|Y_{L_i})$. Therefore, for each winner $v_i\in X_L$, we have
\begin{equation}
	p_i=\left
	\{\begin{IEEEeqnarraybox}[\relax][c]{l's}
		\max\{b'_i\},\text{ if } \sum\nolimits_{v_{i_j}\in Y_{L_i}}b_{i_j}+b_i>B\\
		\max\{\max\{b'_i\},A(v_i|Y_{L_i})\},\text{ else}
	\end{IEEEeqnarraybox}
	\right.
\end{equation}
\noindent
Based on the (14), a truthful budgeted selection and pricing (TBSAP) algorithm is shown in Algorithm \ref{a2}.

\begin{lem}
	The TBSAP is individually rational.
\end{lem}
\begin{proof}
	Let $v_{i_i}$ be the vehicle $v_i$'s replacement which appears in the $i$-th position in the sorting (12) over $V_{-i}$. Because the vehicle $v_{i_i}$ cannot be in the $i$-th position when the winner $v_i$ joins in this sorting. Thus, if $i\leq L_i$, we have $\hat{P}(v_i|Y_{i-1})\geq\hat{P}(v_{i_i}|Y_{i-1})$ which results in $b_i\leq b_{i_i}\cdot A(v_i|Y_{i-1})/A(v_{i_i}|Y_{i-1})\leq p_i$. If $i>L_i$ or $i$ does not exist, we have $b_i\leq A(v_{i}|X_{i-1})$ since the $v_i$ is a winner, and $b_i\leq A(v_{i}|X_{i-1})=A(v_{i}|Y_{i-1})\leq A(v_i|Y_{L_i})\leq p_i$ because of the submodularity.
\end{proof}

\begin{lem}
	The TBSAP is profitable.
\end{lem}
\begin{proof}
	Recall that the vehicle $v_L$ is the last one that satisfies $A(v_L|X_{L-1})\geq b_L$ in the sorting (11), we have the profit $P(X_L)=\sum_{1\leq i\leq L}A(v_i|X_{i-1})-\sum_{1\leq i\leq L}p_i$ according to the (2), which is sufficient to show that $A(v_i|X_{i-1})\geq p_i$ for each $1\leq i\leq L$. For each vehicle $v_i\in X_L$, we consider the two sub-cases shown as follows:
	
	\noindent
	\textbf{Case 1: } Recall that the vehicle $v_{i_{L_i}}$ is the last one that satisfies $\hat{P}(v_{i_{L_i}}|Y_{L_i-1})\geq 0$ and $\sum_{v_{i_j}\in Y_{L_i-1}}b_{i_j}+b_i\leq B$ in the sorting (12), we have $p_i=\max\{b'_i\}$ if $\sum_{v_{i_j}\in Y_{L_i}}b_{i_j}+b_i>B$. Thus, we denote by
	\begin{equation}
		k=\arg\max_{1\leq j\leq L_i}\left\{b'_{i(1)},\cdots,b'_{i(j)},\cdots,b'_{i(L_i)}\right\}
	\end{equation}
	Based on that, we have
	\begin{flalign}
		p_i&=b_{i_k}\cdot{A(v_i|Y_{k-1})}/{A(v_{i_{k}}|Y_{k-1})}\leq A(v_i|Y_{k-1})\\
		&\leq A(v_i|X_{i-1})
	\end{flalign}
	where the inequality (16) is because we have ${A(v_{i_{k}}|Y_{k-1})}\geq b_{i_k}$. From the Lemma 2, we have $b_i\leq p_i$. Consider a vehicle $v_i\in X_L$ with its bid $b_i$, it cannot be moved forward in the sorting (11) by increasing its bid. We can know that $i\leq k$ and $X_{i-1}\subseteq Y_{k-1}$ due to the fact that $X_j=Y_j$ when $j<i$. Therefore, we have $A(v_i|Y_{k-1})\leq A(v_i|X_{i-1})$ because of its submodularity, and the inequality (17) is established.
	
	\noindent
	\textbf{Case 2: }Otherwise, we have $p_i=\max\{\max\{b'_i\},A(v_i|Y_{L_i})\}$ if $\sum_{v_{i_j}\in Y_{L_i}}b_{i_j}+b_i\leq B$. Thus,
	\begin{equation}
		p_i\leq A(v_i|Y_{L_i})\leq A(v_i|X_{i-1})
	\end{equation}
	due to the similar analysis with the inequality (17). Combining the case 1 and case 2, this lemma can be proven.
\end{proof}

\begin{lem}
	The TBSAP is truthful.
\end{lem}
\begin{proof}
	According to the Theorem 1, it is sufficient to show that our budgeted reverse auction mechanism satisfies the monotonicity and critical payment. Let us look at the monotonicity first. Consider a vehicle $v_i\in X_L$ with its bid $b_i$, it cannot be moved backward in the sorting (11) by offer a lower bid $b^\circ_i$ with $b^\circ_i<b_i$ since $A(v_i|X_{i-1})/b^\circ_i>A(v_i|X_{i-1})/b_i$. Therefore, if $v_i$ is a winner by offering a bid $b_i$ in the winner allocation stage, it must be a winner by offering $b^\circ_i$.
	
	Then, we show the payment $p_i$ is the critical payment for the vehicle $v_i$ where it cannot win this auction if giving a bid larger than $p_i$ definitely. Based on the (14), when $\sum_{v_{i_j}\in Y_{L_i}}b_{i_j}+b_i>B$, the $v_i$ must be able to replace one vehicle in $Y_{L_i}$ if it wants to be a winner. If the $v_i$ offers a bid $b_i^\circ$ larger than $p_i=\max\{b'_i\}$, we have 
	\begin{equation}
		A(v_i|Y_{j-1})/b_i^\circ<A(v_{i_j}|Y_{j-1})/b_{i_j}\text{ for }1\leq j\leq L_i
	\end{equation}
	The $v_i$ cannot replace any one in $Y_{L_i}$ with bid $b_i^\circ$, thereby it cannot be a winner because of the budget constraint. Based on the (14), when $\sum_{v_{i_j}\in Y_{L_i}}b_{i_j}+b_i\leq B$, suppose the $v_i$ can not replace any vehicle in $Y_{L_i}$, it sill can be a winner if $A(v_i|Y_{L_i})\geq b_i$ since there is residual budget. If the $v_i$ offers a bid $b_i^\circ$ larger than $p_i=\max\{\max\{b'_i\},A(v_i|Y_{L_i})\}$, it not only fail to replace any one in $Y_{L_i}$ according to the (19), but also have $b_i^\circ>A(v_i|Y_{L_i})$, thereby it cannot be a winner. 
\end{proof}

\begin{lem}
	The TBSAP is computationally efficient.
\end{lem}
\begin{proof}
	We have known that there are at most $B/\min_{v_i\in V}\{b_i\}$ winners and it takes $O(Bnm/\min_{v_i\in V}\{b_i\})$ to finish the winner allocation stage where $|T|=m$ and $|V|=n$. Then, for each winner $v_i\in X_L$, a process similar to the winner allocation stage needs to be executed that takes $O(Bnm/\min_{v_i\in V}\{b_i\})$ running time as well. Thus, the running time of the payment determination state is $O(B^2nm/(\min_{v_i\in V}\{b_i\})^2)$ which is the total time complexity as well. The TBSAP can be done in polynomial time.
\end{proof}

\begin{thm}
	The TBSAP, shown as Algorithm \ref{a2}, is an effective budgeted reverse auction mechanism to solve Problem 1 that satisfies individual rationality, profitability, truthfulness, and computational efficiency.
\end{thm}
\begin{proof}
	As indicated above, this theorem can be proved by putting from Lemma 2 to Lemma 5 together.
\end{proof}

\section{Numerical Simulations}
In this section, we first evaluate the performance of our reputation-based DPoS mechanism, then test the correctness and efficiency of our budgeted reverse auction algorithm. The simulation setup and results will be displayed.

\subsection{Simulation Setup}
To evaluate the reputation-based DPoS mechanism, we need to observe how different behaviors in a consensus round affects its reputation value. Here, we take two full nodes as an example to demonstrate it, where one is a normal node that behaves legitimately and the other is an abnormal node that sometimes makes some wrong behaviors, such as not voting in witness election, producing an invalid block as the leader, or verifying a block wrongly. Shown as the (4), we give the parameters $A=0.005$, $B=0.05$, and $C=0.01$. Then, we compare our reputation-based DPoS with the general DPoS mechanism whose each full node in $\mathbb{Z}$ is given by the same voting weight. In other words, we do not consider their reputation values in witness election of the general DPoS. In these two DPoS mechanism, each normal full node $TA_{a}$ votes to support $TA_{a'}\in\mathbb{Z}\backslash\{TA_a\}$ with $Rep_{a'}\geq\theta$ where $\theta=0.5$ is a threshold. Conversely, we consider the most extreme case where each abnormal full node $TA_{a}$ votes to support $TA_{a'}\in\mathbb{Z}\backslash\{TA_a\}$ with $Rep_{a'}<\theta$. In our BRTM system, we define the full node set with $|\mathbb{Z}|=100$ and the witness committee with $|\mathbb{M}|=70$. To simulate a real state, we give the reputation values of normal full nodes by sampling from $[0.5,1]$ uniformly and abnormal full nodes by sampling from $[0,0.5)$ uniformly. To evaluate the reliability of this system, we define the ratio of abnormal full nodes (RAFN) as ``$[\#\text{ abnormal full nodes}]/|\mathbb{Z}|$'' and the ratio of normal witnesses (RNW) as $[\#\text{ normal witnesses}/\mathbb{M}]$.

To simulate the budgeted reverse auction mechanism, we consider a smart city $S=(TA,V)$ with a $1000m\times1000m$ square area and the task set $T$ issued by the $TA$ are distributed over this area. Then, the active vehicles $V$ are distributed over this area arbitrarily as well. For each vehicle $v_i\in V$, its ability to detect is different, thereby we assume there is a detection distance $d_i$ that is distributed in $[10,30]$ uniformly. It means that the $v_i$'s task set $T_i$ contains all tasks whose distances from the $v_i$ are less than $d_i$. Finally, the TA's appraisement $a_j$ for each task $t_j\in T$ is distributed in $(0,10]$ uniformly and the cost of $T_i$ for the $v_i$ can be denoted by $c_i=\kappa\cdot|T_i|$ where the parameter $\kappa$ is distributed in $(0,5]$ uniformly.

\begin{figure}[!t]
	\centering
	\includegraphics[width=2.5in]{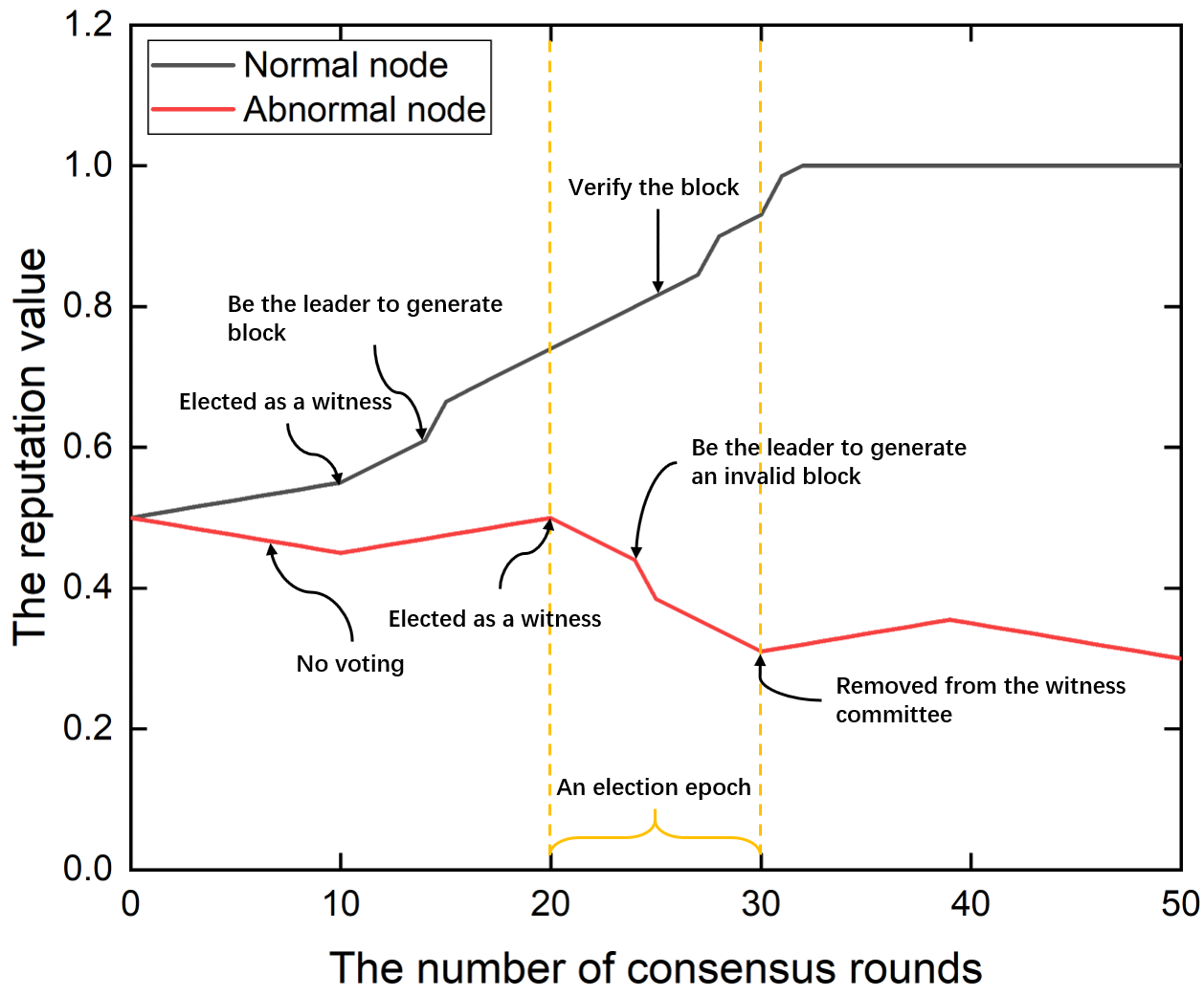}
	\caption{The reputation values change with different behaviors of the two full nodes in the consensus process.}
	\label{fig3}
\end{figure}

\begin{figure}[!t]
	\centering
	\includegraphics[width=2.5in]{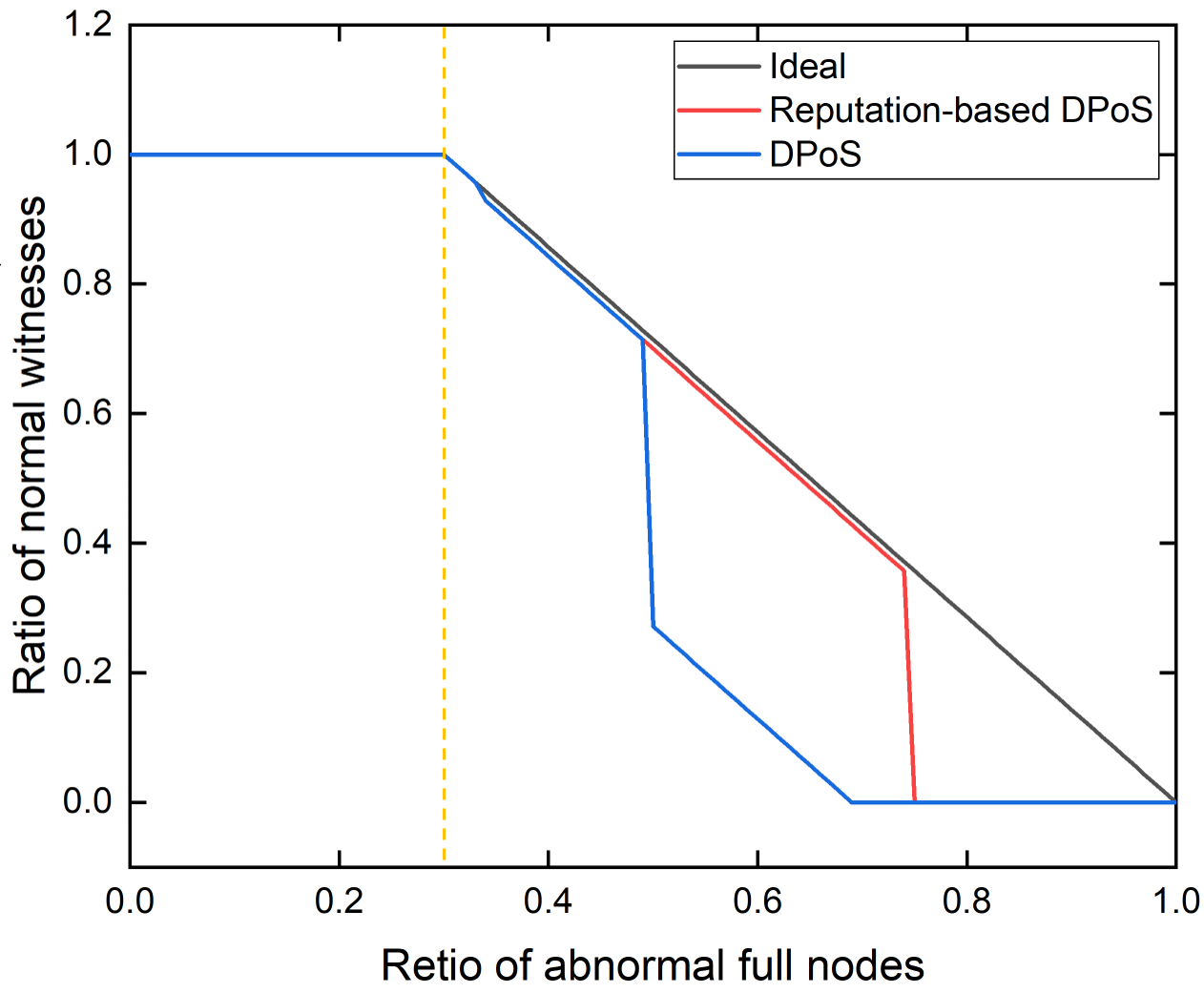}
	\caption{The ratio of normal witnesses (RNW) changes with the ratio of abnormal full nodes (RAFN).}
	\label{fig4}
\end{figure}

\subsection{Simulation Results and analysis}
The simulation results about our reputation-based DPoS mechanism are shown in Fig. \ref{fig3} and Fig. \ref{fig4}.

\textbf{1) The impact of behavior on reputation: }Fig. \ref{fig3} draws the reputation values change with different behaviors of the two full nodes the consensus process. Shown as Fig. \ref{fig3}, we can see that the repuation of the normal node is increased gradually by its legitimate behaviors, but the reputation of the abnormal node is decreased by its wrong behaviors. Here, we assume there are $10$ consensus rounds for each election epoch, namely $|\mathbb{D}|=10$. To the normal node, it votes to elect witnesses at the first election epoch ($1$-th to $10$-th round), which adds $0.005$ reputation per round. At the second epoch ($11$-th to $20$-th round), it is selected as a witness member because of its good reputation. Thus, it votes and verifies the block correctly, which adds $0.015$ reputation per round. Especially, at the $15$-th consensus round, it becomes the leader that generates a block successfully, which adds $0.055$ reputation in this round. These legitimate behaviors increase its reputation until approaching $1$. To the abnormal node, it is selected as a witness member at the third epoch ($21$-th to $30$-th round). It does not vote and verify the block, which reduces $0.015$ reputation per round. Especially, at the $15$-th consensus round, it becomes the leader that generates a block successfully, which adds $0.055$ reputation in this round. Especially, at the $25$-th consensus round, it becomes the leader that generates an invalid block, which reduces $0.055$ reputation in this round.

\begin{figure}[!t]
	\centering
	\includegraphics[width=2.5in]{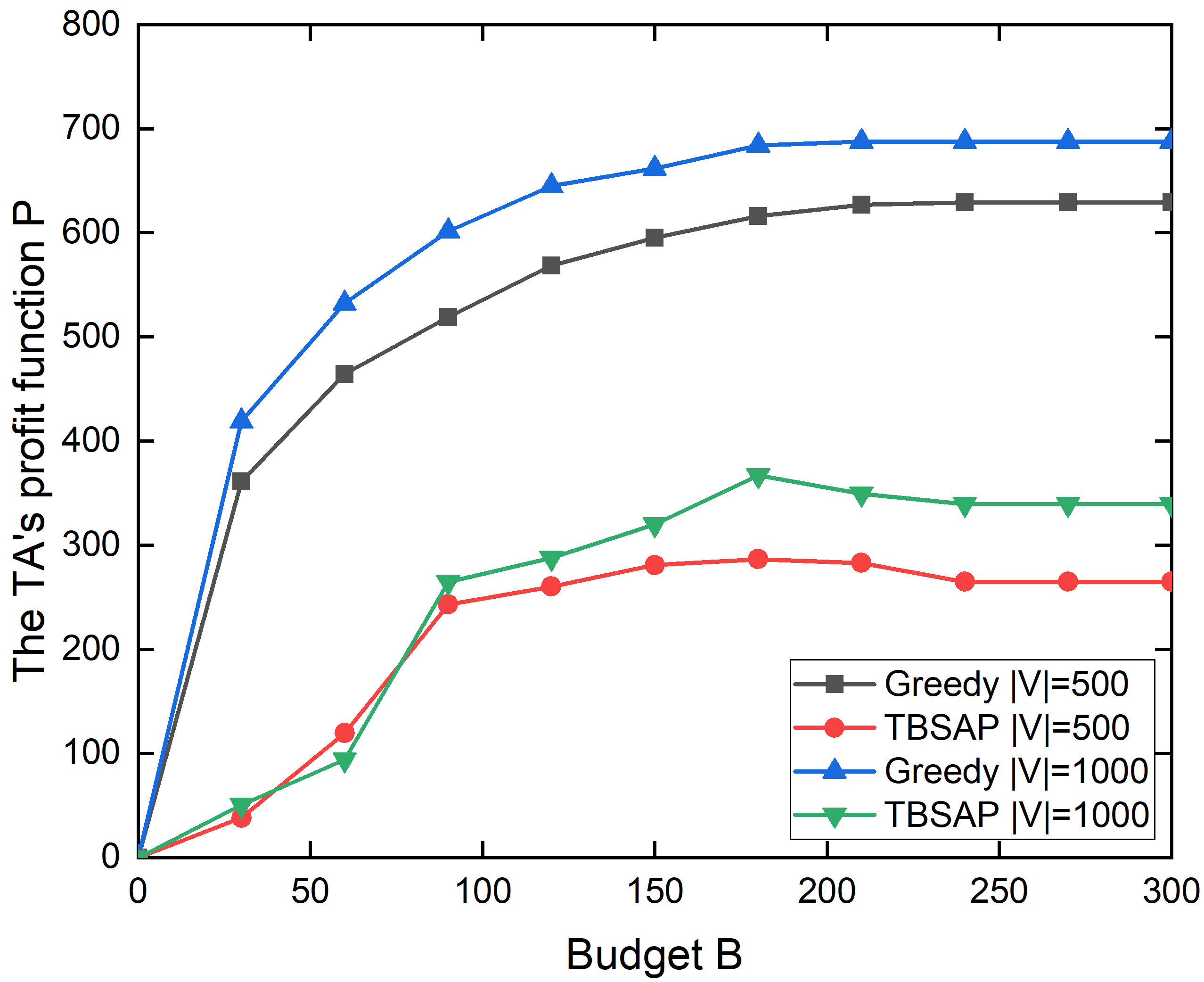}
	\caption{The performance obtained by greedy-heuristic and TBSAP algorithm under the different budgets and number of vehicles.}
	\label{fig5}
\end{figure}

\begin{figure}[!t]
	\centering
	\includegraphics[width=2.5in]{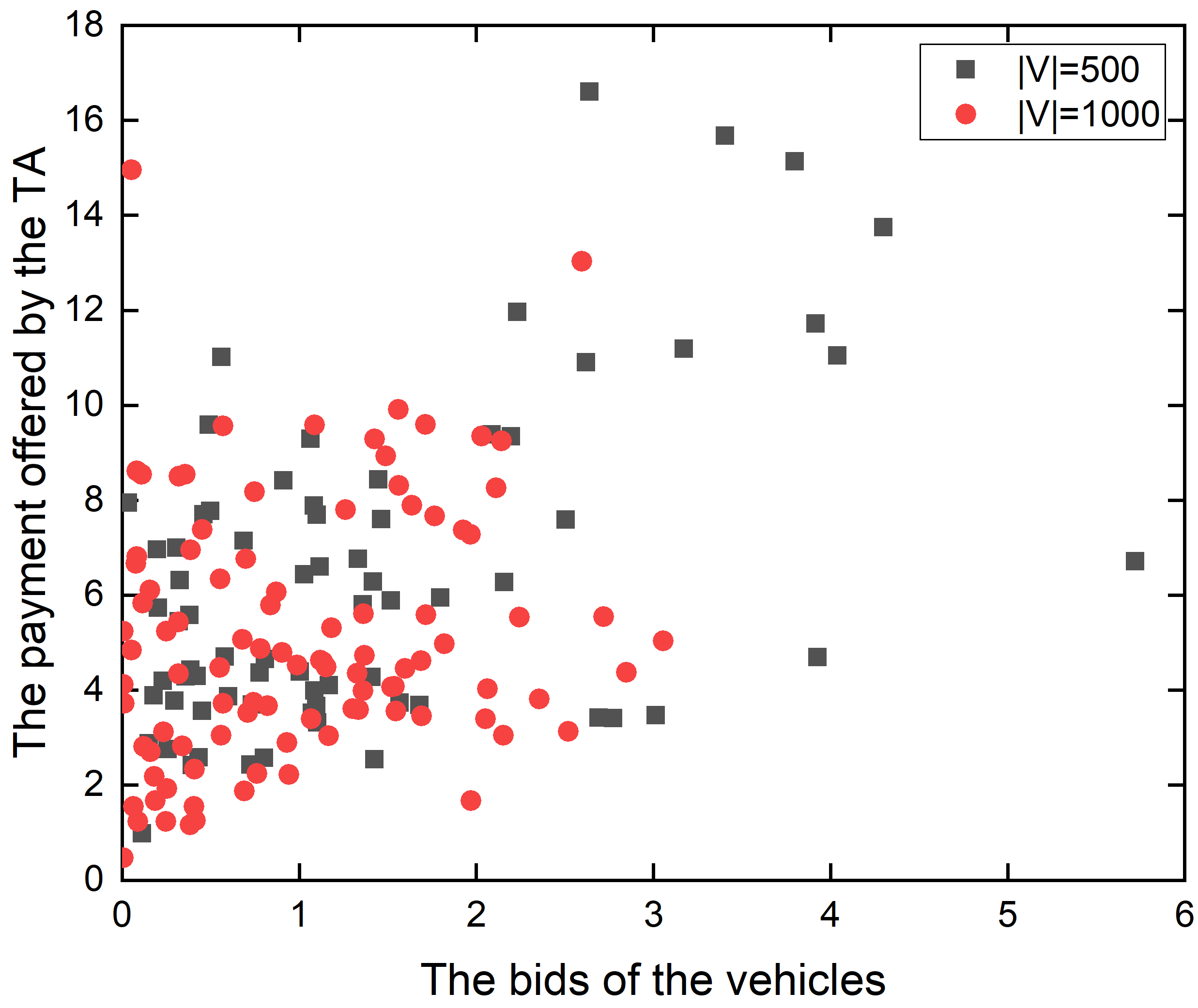}
	\caption{The distributions of bid and payment for each vehicle under the different number of vehicles with budget $B=100$}
	\label{fig6}
\end{figure}

\textbf{2) Reliability: }Fig. \ref{fig4} draws the ratio of normal witnesses changes with the ratio of abnormal full nodes, which describes the reliability of our BRTM system. The higher the ratio of normal witnesses is, the more reliable this system is. Shown as Fig. \ref{fig4}, the ideal line is the maximum RNW we can obtain theoretically. For instance, if RAFN is $0.6$, namely there are $60$ abnormal full nodes, the optimal RNW is equal to $100\cdot(1-0.6)/70$. The RNWs are all closed to the ideal value when the RAFN is less than $0.5$. This is because the number of abnormal nodes is less than normal nodes, thereby their voting hardly changes the outcome of witness election. However, when the RAFN is from $0.5$ to $0.75$, the RNW of our reputation-based DPoS is larger than that of general DPoS obviously since the normal nodes have larger voting weight even though their quantity is small in total. Therefore, the reliability of our system is improved by the reputation-based DPoS especially when the proportion of abnormal nodes is higher and the most extreme case happens where each abnormal node votes to support all abnormal nodes.

The simulation results about our budgeted reverse auction mechanism are shown in Fig. \ref{fig5} and Fig. \ref{fig6}.

\textbf{3) The TA's profit function: }Fig. \ref{fig5} draws the performance obtained by greedy-heuristic and TBSAP algorithm under the different budgets and number of vehicles in this city, where the number of vehicles $|V|$ is $500$ or $1000$. Shown as Fig. \ref{fig5}, we can observe that the TA's profit obtained by TBSAP is less than that obtained by greedy-heuristic under any budget in order to ensure the truthfulness. To the results obtained by greedy-heuristic, we can observe the TA's profit increases as the budget increases. When the budget is larger than $200$, the TA's profit keeps unchanged. There is a similar evolutive trend in the results obtained by TBSAP, but it declines slightly when the budget is larger than $200$. At this time, none of the  optional vehicles has a positive unit marginal gain, thereby the operation step in line 25 of Algorithm \ref{a2} is possible to increase the payment and reduce the profit further.

\textbf{4) Bids and payments: }Fig. \ref{fig6} draws the distributions of bid and payment for each vehicle under the different number of vehicles with a budget $B=100$. Shown as Fig. \ref{fig6}, we can see that the points are more concentrated on the bottom left which have lower bids and payments when the number of vehicles is larger. This is because there are more active vehicles with lower bids that can be selected to undertake collection tasks. It explains why the TA's profit increases with the number of vehicles in this city.

\section{Conclusion}
In this paper, we designed and implemented a reliable and efficient traffic monitoring system based on blockchain technology and budgeted reverse auction mechanism. To enhance the security and reliability of this system, we gave a lightweight information trading framework by using asymmetric encryption. We devised a reputation-based DPoS consensus mechanism so as to improve the efficiency of recording and storing in blockchain. Then, to incentivize vehicles to undertake collection tasks, we developed a budgeted reverse auction algorithm that satisfies individual rationality, profitability, truthfulness, and computational efficiency. Finally, the results of numerical simulations indicated that our model is valid, and verify the correctness and efficiency of our algorithms.

\section*{Acknowledgment}

This work is partly supported by National Science Foundation under grant 1747818 and 1907472.

\ifCLASSOPTIONcaptionsoff
  \newpage
\fi

\bibliographystyle{IEEEtran}
\bibliography{references}

\begin{IEEEbiography}[{\includegraphics[width=1in,height=1.25in,clip,keepaspectratio]{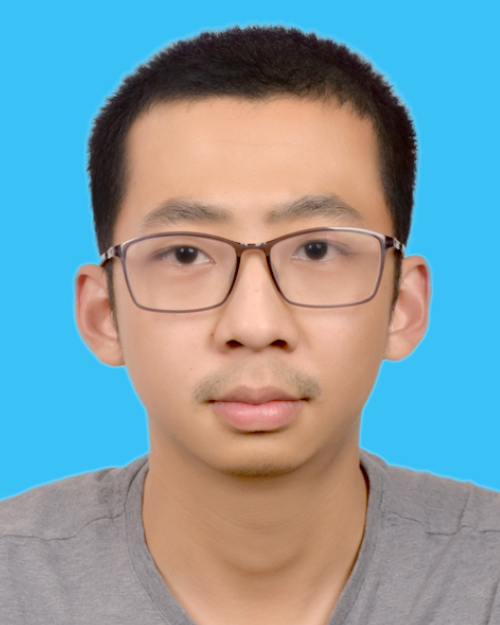}}]{Jianxiong Guo}
	is a Ph.D. candidate in the Department of Computer Science at the University of Texas at Dallas. He received his B.S. degree in Energy Engineering and Automation from South China University of Technology in 2015 and M.S. degree in Chemical Engineering from University of Pittsburgh in 2016. His research interests include social networks, data mining, IoT application, blockchain, and combinatorial optimization.
\end{IEEEbiography}

\begin{IEEEbiography}[{\includegraphics[width=1in,height=1.25in,clip,keepaspectratio]{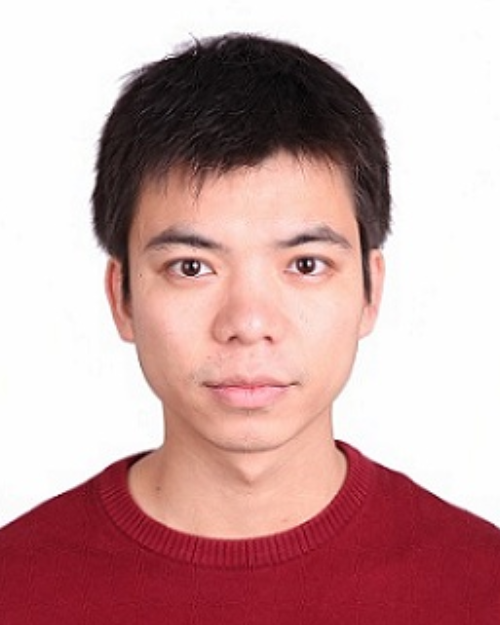}}]{Xingjian Ding}
	received the BE degree in electronic information engineering from Sichuan University, Sichuan, China, in 2012. He received the M.S. degree in software engineering from Beijing Forestry University, Beijing, China, in 2017. Currently, he is working toward the PhD degree in the School of Information, Renmin University of China, Beijing, China. His research interests include wireless rechargeable sensor networks algorithm, design and analysis, and blockchain.
\end{IEEEbiography}

\begin{IEEEbiography}[{\includegraphics[width=1in,height=1.25in,clip,keepaspectratio]{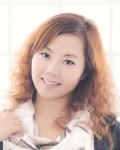}}]{Weili Wu}
	received the Ph.D. and M.S. degrees from the Department of Computer Science, University of Minnesota, Minneapolis, MN, USA, in 2002 and 1998, respectively. She is currently a Full Professor with the Department of Computer Science, The University of Texas at Dallas, Richardson, TX, USA. Her research mainly deals in the general research area of data communication and data management. Her research focuses on the design and analysis of algorithms for optimization problems that occur in wireless networking environments and various database systems.
\end{IEEEbiography}

\end{document}